\documentclass[a4paper,11pt,reqno]{amsart} 
\usepackage{amsmath}
\usepackage{amsfonts}
\usepackage[T1]{fontenc}  
\usepackage[utf8]{inputenc}
\usepackage[mathscr]{euscript}
\usepackage{verbatim}
\usepackage{amssymb,latexsym}
\usepackage{amsthm}
\usepackage{color}
\usepackage{graphicx}
\usepackage[hidelinks]{hyperref}
\usepackage[font=scriptsize]{caption}
\usepackage{bbm}
\usepackage{cite}
\usepackage{amsmath,accents}
\usepackage{subcaption}
\usepackage{float}

\newcommand{\Ab}{\mathbf{A}}
\newcommand{\Eb}{\mathbf E}

\usepackage[lmargin=3.5 cm,rmargin=3.5 cm,tmargin=3.5cm,bmargin=2.5cm,paper=a4paper]{geometry}
\DeclareMathOperator{\curl}{curl}\DeclareMathOperator{\Div}{div}
 \DeclareMathOperator{\dist}{dist} 
\DeclareMathOperator{\supp}{supp} \DeclareMathOperator{\dom}{\mathrm {Dom}}
 
\DeclareMathOperator{\re}{\mathrm{Re}}

\DeclareMathOperator{\spc}{\mathrm{sp}}

\newtheorem{thm}{Theorem}[section]

\newtheorem{theorem}[thm]{Theorem}
\newtheorem{assumption}[thm]{Assumption}
\newtheorem{notation}[thm]{Notation}

\newtheorem{lemma}[thm]{Lemma}
\newtheorem{proposition}[thm]{Proposition}

\theoremstyle{remark}
\newtheorem{rem}[thm]{Remark}
\newtheorem*{remark}{Remark}

\newcommand{\nb}{\nabla}
\newcommand{\R}{\mathbb{R}}
\newcommand{\Fb}{\mathbf{F}}

\newcommand{\N}{\mathbb{N}}
\newcommand{\C}{\mathbb{C}}

\newcommand{\Hd}{H_{{\rm div}}^1(\Omega)}
\newcommand{\Om}{\Omega}
\newcommand{\kp}{\kappa}
\newcommand{\kn}{\nabla-i\kappa H{\bf A}}
\newcommand{\Es}{{\rm E}_{\rm g.st}(\kappa,  H)}
\newcommand{\GL}{\mathcal E_{\kappa,H}}

\def\sig#1{\vbox{\hsize=5.5cm
		\kern2cm\hrule\kern1ex
		\hbox to \hsize{\strut\hfil #1 \hfil}}}
\newcommand\signatures[4]{%
	\vspace{3cm}
	\hbox to \hsize{\hfil #1, \today\hfil}
	\vspace{3cm}
	\hbox to \hsize{\quad#2\hfil\hfil #3\quad}
	\vspace{3cm}
	\hbox to \hsize{\hfil#4\hfil}}
\numberwithin{equation}{section}

\title{Magnetic steps on the threshold of the normal state}
\author[W. Assaad]{Wafaa Assaad}
\address{Lund University, Department of Mathematics, Lund, Sweden}
\begin{document}
	\maketitle
	\begin{abstract}
		Superconductivity in the presence of a step magnetic field has been recently the focus of many works. This contribution examines the behavior of a two-dimensional superconducting domain, when superconductivity is lost in the whole domain except near the intersection points of the discontinuity edge and the boundary. The problem involves its own effective energy.  We provide local estimates of the minimizers in neighbourhoods of the intersection points. Consequently, we introduce new critical fields marking the loss of superconductivity in the vicinity of these points. The study is modelled by the Ginzburg--Landau theory, and large Ginzburg--Landau parameters are considered.
		
		~
		
		\medskip
		
	\end{abstract}

	\section{Introduction}\label{sec:int}

Hundreds of contributions have investigated the response of a type-II superconductor with a large Ginzburg--Landau (GL) parameter to applied magnetic fields (see the two monographs~\cite{sandier2007vortices,fournais2010spectral} in the mathematical literature). In many generic situations in two- and three-dimensional domains submitted to smooth magnetic fields, it is shown that superconductivity eventually breaks down under an increasing magnetic field~\cite{saint1963onset,giorgi2002breakdown,lu2000gauge,helffer2001magnetic,helffer2003upper,fournais2010spectral}; the superconductor is said to pass to the normal state. In these situations, the transition to the normal state occurs at a unique value of the applied field's intensity\footnote{We say that the transition is monotone. For counterexamples of such a monotonicity, see~e.g.~\cite{little1962observation,erdHos1997dia,fournais2015lack,helffer2019thin,kachmar2019counterexample,kachmar2019superconductivity}.}, called the third critical field and depending on the GL parameter. The last phase preceding such a transition has been extensively studied for two-dimensional domains with piecewise smooth boundary (having possibly a finite number of corners) submitted to \emph{smooth} magnetic fields (see e.g.~\cite{jadallah2001onset,pan2002upper,pan2002schrodinger,bonnaillie2005fundamental,bonnaillie2006asymptotics,bonnaillie2007superconductivity,fournais2010spectral,helffer2018density,correggi2019effects}).

Here, we study the aforementioned phase in the case of a certain \emph{discontinuous}  magnetic field. In this situation, we provide precise information about the superconductivity localization, right before its breakdown. Such a localization was suggested but not proven in the recent paper~\cite{Assaad3}.
\subsection{The functional}
The problem is modelled by the GL theory. We consider a cross section,  $\Om\subset\R^2$, of an infinite cylindrical wire subjected to a magnetic field, whose direction is  parallel to the axis of the cylinder and whose profile (the scalar magnetic field) is the function $B_0\in L^2(\Om;[-1,1])$. $\Om$ is assumed to be open,  bounded with a smooth boundary, and simply connected. The GL free energy is given by the functional:
\begin{equation}\label{eq:GL}
\GL (\psi,\Ab)= \int_\Om \Big( \big|(\nb-i\kp H {\mathbf
	A})\psi\big|^2-\kp^2|\psi|^2+\frac{\kappa^2}{2}|\psi|^4 \Big)\,dx
+\kp^2H^2\int_{\Om}\big|\curl\Ab-B_0\big|^2\,dx,
\end{equation}
with $\psi \in H^1(\Om;\C)$ and  $\Ab\in H^1(\Om;\R^2)$. 
$\psi$ is the order parameter with $|\psi|^2$ being a measure of the Cooper pair electrons density, and
$\Ab$ is the vector potential whose $\curl$ represents the induced magnetic field in the sample. $\kp>0$ is a characteristic scale of the sample called  the GL parameter, assumed to be large ($\kappa\to+\infty$), which corresponds to extreme type-II superconductors in physics. Finally, $H>0$ is the intensity of the applied magnetic field. 

The functional in~\eqref{eq:GL} admits a gauge invariance property\footnote{The physically relevant quantities $|\psi|^2$, $\curl \Ab$ and $|(\nabla-i\kappa H\Ab)\psi|^2$ are invariant
	under the  transformation $(\psi,\Ab)\mapsto (e^{i\varphi\kappa H}\psi,\Ab+\nb \varphi)$ for any  $\varphi \in H^2(\Om;\R)$.}. 	Hence, one may restrict the energy minimization with respect to $(\psi,\Ab)$ (originally done in the space $H^1(\Om;\C)\times H^1(\Om;\R^2)$) to the space $H^1(\Om;\C)\times\Hd$, where 
\begin{equation*}\label{eq:Hd}
\Hd=
\left\{
\Ab \in H^1(\Om;\R^2)~:~ \Div\Ab=0 \ \mathrm{in}\ \Om,\ \Ab\cdot\nu=0\ \mathrm{on}\ \partial \Om
\right\}
\end{equation*}
and $\nu$ is a unit normal vector of $\partial \Omega$.
We define the following ground-state energy
\begin{equation} \label{eq:gr_st}
\Es=\inf\{\GL(\psi,\Ab)~:~(\psi,\Ab) \in H^1(\Om;\C)\times\Hd\}.
\end{equation}

Critical points $(\psi, \Ab) \in H^1(\Om;\C)\times\Hd $ of $\GL$ are weak solutions of the following GL equations:
\begin{equation}\label{eq:Euler}
\begin{cases}
\big(\kn\big)^2\psi=\kp^2(|\psi|^2-1)\psi &\mathrm{in}\ \Om,\\
-\nb^{\perp}  \big(\curl\Ab-B_0\big)= \frac{1}{\kp H}\mathrm{Im}\big(\overline{\psi}(\nb-i\kp H \Ab)\psi\big) & \mathrm{in}\ \Om,\\
\nu\cdot(\kn)\psi=0 & \mathrm{on}\ \partial \Om,\\
\curl\Ab=B_0 & \mathrm{on}~ \partial \Om,
\end{cases}
\end{equation}
with $\nb^\perp= (\partial_{x_2},-\partial_{x_1})$.

\subsection{Literature summary and new contribution}
The scenarios occurring in the case of smooth applied magnetic fields are well-known in the literature (see e.g.~\cite{lu1999estimates,helffer2001magnetic,pan2002schrodinger,sandier2007vortices,raymond2009sharp,fournais2010spectral,fournais2011nucleation,Helffer,attar2015energy,attar2015ground,correggi2016boundary,correggi2016effects,correggi2017surface,dauge2018semiclassical,fournais2019concentration}). 

In particular, in domains with smooth boundaries submitted to uniform magnetic fields, three subsequent transitions are observed while increasing the intensity of the applied field, and correspondingly three critical values of the field's intensity---critical fields---are considered:
\begin{itemize}
	\item The first critical field $H^{\rm unif}_{C_1}(\kappa)=\mathcal O(\kappa^{-1}\log(\kappa))$: indicating the transition between the perfect superconductivity state, where the whole sample is superconducting, and the state of vortices nucleation.
	\item The second critical field $H^{\rm unif}_{C_2}(\kappa)=\kappa$: indicating the transition between the bulk and surface superconductivity states. In the surface state, superconductivity is exclusively and~\emph{uniformly} distributed along the boundary.
	\item The third critical field $H^{\rm unif}_{C_3}(\kappa)=\mathcal O(\Theta_0^{-1}\kappa)$, where $\Theta_0\in(0,1)$ is introduced later: indicating the transition between the surface and normal states.
\end{itemize}

This picture is modified if the boundary of the sample admits singularities, like corners. In such situations, when the  magnetic field is uniform, the bulk phase is not essentially altered, but the surface and normal phases prove to be affected by the presence of corners~\cite{jadallah2001onset,pan2002upper,bonnaillie2007superconductivity,correggi2017surface,helffer2018density,correggi2019effects}. In particular, under a certain spectral/geometric assumption (see~\cite[Remark~1.1 \& Assumption~1.3]{bonnaillie2007superconductivity}), the third critical field experiences a shift up in the presence of corners ($H^{\rm cor}_{C_3}(\kappa)>H^{\rm unif}_{C_3}(\kappa)$), although is still of order $\kappa$. In addition, a new sub-phase appears in the surface superconductivity phase, when $H$ exceeds a certain critical value $H_{\rm int}^{\rm cor}(\kappa)$. There, superconductivity is no longer uniformly distributed along the boundary, but solely localized at the corners. Moreover, superconductivity disappears at a corner, $\mathsf s_j$, once the field's intensity exceeds certain critical field $H^{\rm cor}_{j}(\kappa)\in \big(H_{\rm int}^{\rm cor}(\kappa),H^{\rm cor}_{C_3}(\kappa)\big)$, depending on the opening angle of this corner (see e.g.~\cite{bonnaillie2007superconductivity, helffer2018density}).

Discontinuous magnetic fields were treated for the first time in the context of the non-linear GL functional in~\eqref{eq:GL},  in the recent works~\cite{Assaad, Assaad2019,Assaad3}. These works considered a magnetic field $B_0$ which is a step function, satisfying the following assumptions (see~Figure~\ref{fig1}):
\begin{assumption}\label{assump1}~
	\begin{enumerate}
		\item $\Omega_1$ and $\Omega_2$ are two disjoint open sets.
		\item $\Omega_1$ and $\Omega_2$ have a finite number of connected components.
		\item $\partial\Omega_1$ and $\partial\Omega_2$ are piecewise smooth with  a finite number of corners.
		\item $\Gamma=\partial\Omega_1\cap\partial\Omega_2$ is the union of a finite number of disjoint simple smooth curves $\{\Gamma_k\}_{k \in \mathcal K}$\,; we will refer to  $\Gamma$ as the \emph{magnetic edge}.
		\item $\Omega=(\Omega_1\cup\Omega_2\cup\Gamma)^\circ$ and  $\partial\Omega$  is smooth.
		\item For any $k \in \mathcal K$, $\Gamma_k$ intersects $\partial \Om$ at two distinct points. This intersection is transversal, i.e. $\mathrm{T}_{\partial \Omega} \times \mathrm{T}_{\Gamma_k} \neq 0$ at the intersection point, where $\mathrm{T}_{\partial \Omega}$ and  $\mathrm{T}_{\Gamma_k}$ are respectively unit tangent vectors of $\partial \Omega$ and $\Gamma_k$.
		\item $B_0={\mathbbm 1}_{\Omega_1}+a{\mathbbm 1}_{\Omega_2}$, where $a \in [-1,1)\setminus\{0\}$ is a given constant.
	\end{enumerate} 
\end{assumption}
 We borrow the following notation from~\cite{Assaad3}:
 \begin{notation}\label{not:alfa}
 	Since $\Gamma \cap \partial \Om$ is finite, we denote by
 	\[\Gamma \cap \partial \Om=\big\{ \mathsf p_j~:~j \in\{1,\cdots,n\}\big\},\]
 	where $n =\mathrm{Card}(\Gamma \cap \partial \Om)$. For all $j \in \{1,\cdots,n\}$, let   $\alpha_j \in(0,\pi)$ be the angle between $\Gamma$ and  $\partial \Om$ at the intersection point $\mathsf p_j$ (measured towards $\Om_1$). 
 \end{notation}
\begin{figure}[H]
	\begin{subfigure}{0.45\linewidth}
		\centering
		\includegraphics[scale=1]{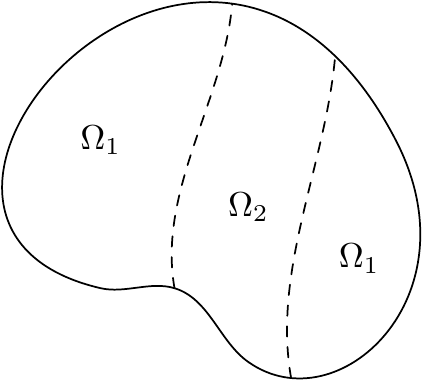}
	\end{subfigure}%
	\begin{subfigure}{0.45\linewidth}
		\centering
		\includegraphics[scale=1]{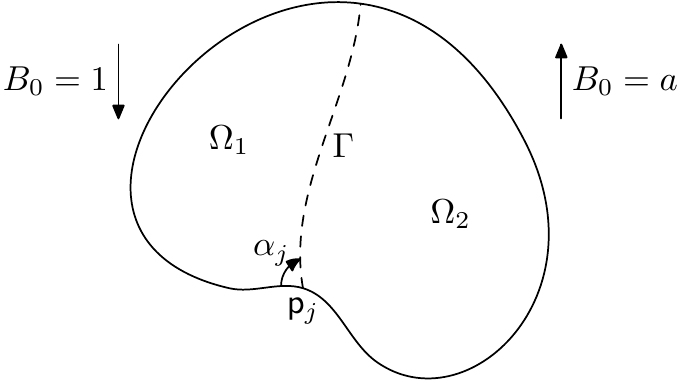}
	\end{subfigure}%
	\caption{Schematic representation of the set $\Om$ subjected to the step magnetic field $B_0$, with the magnetic edge $\Gamma$.} 
	\label{fig1}
\end{figure}
The study of such a discontinuous field case has been mathematically and physically motivated in~\cite{Assaad, Assaad2019,Assaad3}, and more recently in~\cite{assaad2020p4}. 
The latter contribution established the well-posedness of such problems with  step magnetic fields, in the sense that the study outcomes are not essentially affected by a small perturbation of the  field. More precisely, it showed that the ground-state energy in~\eqref{eq:gr_st} with the step magnetic field is the limit of ground-state energies with smooth fields. This interestingly constructed  a bridge between the non-linear GL problems with  discontinuous fields, tackled lately, and those with smooth fields, broadly studied in the literature as mentioned earlier.

\cite{Assaad, Assaad2019,Assaad3} examined superconductivity along an intensity interval extended from the bulk regime to the normal regime. Working under certain spectral conditions  (Assumption~\ref{assump3} below), these three papers introduced the following critical fields:
\begin{itemize}
	\item  $H^{\rm step}_{C_2}(\kappa)=|a|^{-1}\kappa$ ($a$ is the value in Assumption~\ref{assump1}): marking the passage from  bulk superconductivity to superconductivity being  partially/globally distributed along $\Gamma\cup\partial\Om$ (see~\cite[Section~1.5]{Assaad2019}). 
	\item $H_{\rm int}^{\rm step}(\kappa)=(|a|\Theta_0)^{-1}\kappa$: marking the disappearance of superconductivity in the whole sample away from $\Gamma\cap\partial \Om$.
	\item $H^{\rm step}_{C_3}(\kappa)=\mathcal O(\mu_*^{-1}\kappa)$, for a certain $\mu_*<|a|\Theta_0$ introduced later: marking the passage to the normal state.
\end{itemize}
In particular,~\cite{Assaad3} revealed a remarkable similarity between the role of  $\Gamma\cap\partial \Om$, and that of the corners in the foregoing corners situation. Indeed, $\Gamma\cap\partial \Om$ shifts the third critical field to a higher level, compared to that for smooth domains with uniform  fields  ($H^{\rm step}_{C_3}(\kappa)>H^{\rm unif}_{C_3}(\kappa)$). Moreover,  
 a spectral value, $\mu(\alpha_j,a)$, was assigned to each intersection point $\mathsf p_j$ of $\Gamma$ and $\partial \Om$. Then under the condition $\mu(\alpha_j,a)<|a|\Theta_0$ (discussed later),~\cite{Assaad3}  introduced additional fields, $H^{\rm step}_{j}(\kappa):=\big(\mu(\alpha_j,a)\big)^{-1}\kappa$, satisfying $H^{\rm step}_{\rm int}(\kappa)<H^{\rm step}_{j}(\kappa)<H^{\rm step}_{C_3}(\kappa)$. When $H>H^{\rm step}_{j}(\kappa)$,~\cite[Theorem~1.6]{Assaad3}  showed the \emph{non-existence} of superconductivity near the point $\mathsf p_j$. This was reminiscent of the aforementioned corners performance (see~\cite[Section~1.3]{Assaad3} for a more detailed comparison). However, this theorem  did not ensure the \emph{existence} of superconductivity  in the vicinity of $\mathsf p_j$, when $H$ is still below $H^{\rm step}_{j}(\kappa)$. If such an existence is proven, then $H^{\rm step}_{j}(\kappa)$ can be reintroduced as a \emph{critical field} denoting the ending of the superconducting state near $\mathsf p_j$. This will emphasize the  similarity with the corners situation in~\cite{bonnaillie2007superconductivity,helffer2018density}.

\emph{Working under the assumptions in~\cite{Assaad3}}, the current contribution establishes this existence, hence sharpens the results in~\cite[Theorem~1.6]{Assaad3}. More precisely, Theorem~\ref{thm:main} below  provides local estimates of the minimizers which describe the sample's behavior at the threshold of $H^{\rm step}_{j}(\kappa)$, and particularly show the concentration of superconductivity near the point $\mathsf p_j$ when $H<H^{\rm step}_{j}(\kappa)$.

\subsection{Setting and main results}\label{sec:setting}
In what follows, we formally present the setting and the main results of this paper.
In addition to Assumption~\ref{assump1}, we assume the following: 
\begin{assumption}\label{A_2}
	The intensity $H$ depends on $\kappa$ in the following manner
	\begin{equation*}\label{eq:A_2}
	H=b\kappa,
	\end{equation*}
	where $b$ is a  real parameter independent of $H$ and $\kappa$, satisfying 
	\[b > \frac 1{|a|\Theta_0}.\]
\end{assumption} 
\begin{assumption}\label{assump3}~
	 For $j\in\{1,\cdots,n\}$, let $\alpha_j$ be the angle in Notation~\ref{not:alfa}. We assume that $\mu(\alpha_j,a)<|a|\Theta_0$.
\end{assumption}
The values $\Theta_0$ and $\mu(\alpha_j,a)$ appearing in the assumptions above stand for the following quantities:
\begin{itemize}
	\item $\Theta_0\approx 0.59$ is the so-called de Gennes constant---the ground-state energy of the Neumann realization of the Schr\"odinger operator, $-(\nabla+\frac 12 ix^\perp)^2$, with a unit magnetic field  in $\R^2_+$ (see~e.g.~\cite{bonnaillie2007superconductivity}).
	\item $\mu(\alpha,a)$ is the ground-state energy of the Neumann realization of a Schr\"odinger operator with a step magnetic field in $\R^2_+$, introduced in Section~\ref{sec:new_model}. 
\end{itemize}
	\begin{remark}
	The conditions in Assumption~\ref{assump3} are discussed later in the paper (see~Remark~\ref{rem:v_0} and Footnote~\ref{foot}).~\cite[Section~3.3]{Assaad3} provided examples of couples $(\alpha_j,a)$ satisfying this assumption, living in a neighbourhood of $(\pi/2,-1)$.
\end{remark}
The statements of our main results involve the function
\[ E_{\alpha,a}:\big[(|a|\Theta_0)^{-1},+\infty\big) \rightarrow (-\infty,0]\]
defined in Section~\ref{sec:eff}. $E_{\alpha,a}$ is a new effective energy constructed in this paper to deal with our problem near the intersection of the magnetic edge with the boundary. We highlight the following crucial property of this energy:
 \begin{equation}\label{eq:cond2}
 E_{\alpha,a}(b)<0,\ \mbox{for}\ \frac 1{|a|\Theta_0}< b<\frac 1{\mu(\alpha,a)},\ \mbox{and}\
 E_{\alpha,a}(b)=0,\ \mbox{for}\ b\geq\frac 1{\mu(\alpha,a)}.
 \end{equation} 

Let $j\in\{1,\cdots,n\}$ and $\ell\in(0,1)$. For each $\mathsf p_j\in\Gamma\cap \partial \Om$,  we define the set
\begin{equation}\label{eq:Nj}
\mathcal N_j(\ell)=\{x\in \Om,\ \dist(x,\mathsf p_j)\leq\ell\}
\end{equation} 
 We also introduce the set
\begin{equation*}\label{eq:S}
T=\{j\in\{1,\cdots,n\}~:~\mu(\alpha_j,a)<b^{-1}\} 
\end{equation*}
\begin{theorem}\label{thm:main}
	 Let $\rho\in(4/5,1)$. There exists $\kappa_0>1$ and  a function $\mathfrak r:(\kappa_0,+\infty)\to (0,+\infty)$ such that $\lim_{\kappa\to +\infty}\mathfrak r(\kappa)=0$ and the following is true. If $\kappa\geq \kappa_0$ and $(\psi,\Ab) \in H^1(\Om;\C)\times\Hd$ is a minimizer of~\eqref{eq:GL}, then for $\ell\approx\kappa^{-\rho}$ and $j\in T$ we have
		\begin{equation}\label{eq:m2}
	\Big|\kappa^2\int_{\mathcal N_j(\ell)}|\psi|^4\,dx+2 E_{\alpha_j,a}(b)\Big|\leq\mathfrak r(\kappa).
	\end{equation}
	Consequently,
	\begin{equation}\label{eq:m3}
	\Es=\sum_{j\in T} E_{\alpha_j,a}(b)+o(1).
	\end{equation}
\end{theorem}
	The results in Theorem~\ref{thm:main} are actually valid for any $j\in\{1,\cdots,n\}$ (not only for $j\in T$). However in light of the properties of the energies $E_{\alpha_j,a}$ in~\eqref{eq:cond2}, for a \emph{fixed} intensity $H=b\kappa$ (a fixed $b$), superconductivity is negligible at the points $\{\mathsf p_j\}_{j\notin T}$. Hence, the only intersection points which contribute to the global energy, $\Es$, are $\{\mathsf p_j\}_{j\in T}$.

On the other hand, when we increase the intensity $H$ (increasing $b$), we observe successive breakdowns of superconductivity at the points $\mathsf p_j$: Labelling the points $\mathsf p_j$, $j\in\{1,\cdots,n\}$, so that
\[\mu(\alpha_1,a)\geq \mu(\alpha_2,a)\geq\cdots\geq\mu(\alpha_n,a),\]
the fields $H^{\rm step}_{j}(\kappa):=\big(\mu(\alpha_j,a)\big)^{-1}\kappa$, introduced above, satisfy 
\begin{equation}\label{eq:field}
H^{\rm step}_{C_2}(\kappa)< H^{\rm step}_{\rm int}(\kappa)<H^{\rm step}_{1}(\kappa)\leq H^{\rm step}_{2}(\kappa)\leq \cdots\leq H^{\rm step}_{n}(\kappa).\end{equation}
By~\cite{Assaad3}, $H^{\rm step}_{n}(\kappa)$ is the leading-order term of the third critical field\footnote{The value $\mu_*$ mentioned earlier is then $\mu(\alpha_n,a)=\min_j\mu(\alpha_j,a)$.} $H^{\rm step}_{C_3}(\kappa)$.
Thanks to Theorem~\ref{thm:main}, we now consider each  field $H^{\rm step}_{j}(\kappa)$ as a new critical field at which  superconductivity disappears near the intersection point $\mathsf p_j$. 

Consequently,  this paper adds interesting features to the results  in~\cite{Assaad, Assaad2019,Assaad3} that we summarize in Diagram~\ref{fig:diag}.
\begin{figure}[H]
	\centering
	\includegraphics{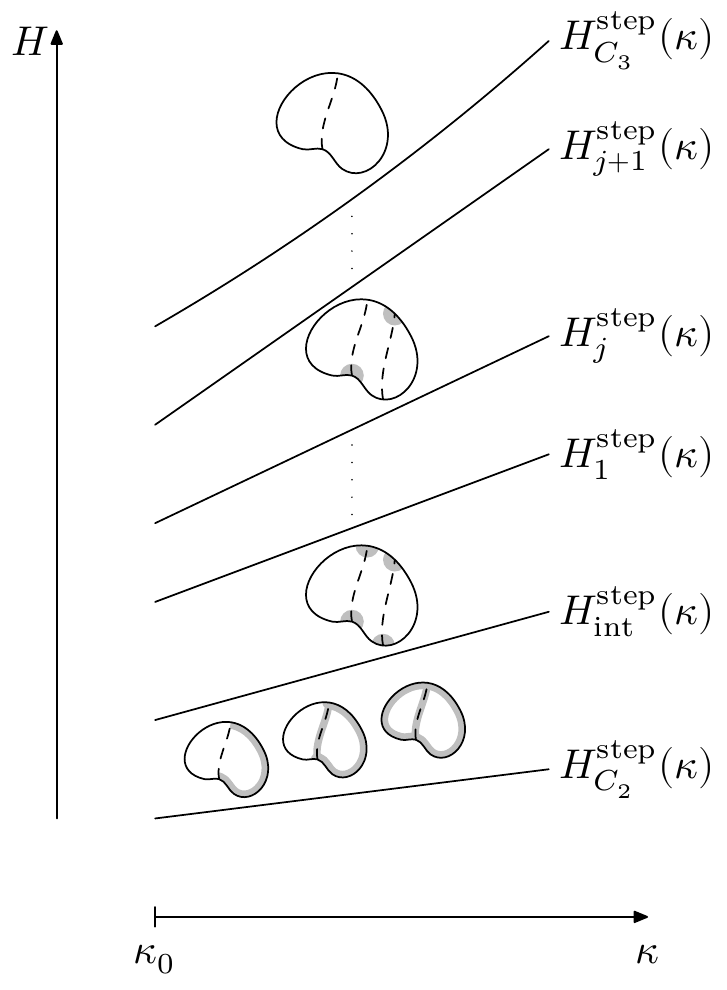}
	\caption{ Schematic phase-diagram representing the distribution of superconductivity in the sample, according to the intensity $H$ of the applied magnetic field and for large values of $\kappa$. Only the grey regions carry superconductivity. The critical lines plotted in the $(\kappa,H)$-plane represent the critical fields in~\eqref{eq:field}. When (and if) $H^{\rm step}_{j}(\kappa)<H<H^{\rm step}_{j+1}(\kappa)$, superconductivity only persists near the points $\mathsf p_k$, for $k\geq j+1$.}
	\label{fig:diag}
\end{figure}
\subsection{Perspectives}
The paper shows that superconductivity is mainly dictated by the intersection of $\Gamma$ and $\partial\Om$, in the intensity regime $\big(H^{\rm step}_{\rm int}(\kappa),H^{\rm step}_{C_3}(\kappa)\big)$ preceding the transition to the normal state. 
Note that this intersection effects do not show up in the leading-order terms of the energy, in the  regime $\big(H^{\rm step}_{C_2}(\kappa),H^{\rm step}_{\rm int}(\kappa)\big)$ where superconductivity is concentrated near $\Gamma\cup\partial \Om$.
 In~\cite[Theorem~1.7]{Assaad2019}, the following estimates of the ground-state energy are established, for  $H\in \big(H^{\rm step}_{C_2}(\kappa),H^{\rm step}_{\rm int}(\kappa)\big)$:
\begin{equation}\label{eq:Eg}
\Es =E^\mathrm{L}_a(b)\kappa +o(\kappa) \qquad(\kappa\to+\infty)\,,
\end{equation}
where
 $b=H/\kappa$ and $E^\mathrm{L}_a(b)$ is a $1$D energy corresponding to the contribution of $\Gamma$ and $\partial\Om$ \emph{away} from their intersection.

However,  predicting a similarity with the corners situation in the \emph{surface} superconductivity regime, we expect the contribution of $\Gamma\cap \partial\Om$ to the remainder terms  in~\eqref{eq:Eg} as follows. \cite{correggi2017surface,correggi2019effects} (see also~\cite{Correggi,correggi2016boundary,correggi2017universal}) have examined the effects of corners on surface superconductivity, in a $2$D domain submitted to a uniform field. Inspired by these works, one may identify an effective model with a step magnetic field on the half-plane, revealing the foregoing contribution of $\Gamma\cap \partial\Om$. This model will be \emph{genuinely} $2$D, unlike the model generating the $1$D energy $E^\mathrm{L}_a(b)$. 

Analogously to~\cite[Theorem~2.1]{correggi2019effects}, we anticipate  the following improved expansion in~\eqref{eq:Eg}: 
\[E^\mathrm{L}_a(b)\kappa -E^1_{\rm corr}\int_0^{\partial \Om_1} k(s)\,ds-E^2_{\rm corr}\int_0^{\partial \Om_2} k(s)\,ds-E^3_{\rm corr}\int_0^{\Gamma} k(s)\,ds-\sum_{j=1}^N E_{{\rm int},\alpha_j,a}+o(1),\]
where $E^k_{\rm corr}$, $k=1,\cdots,3$, are correction energies corresponding to the curvature contribution along $\Gamma$ and $\partial \Om$, and $E_{{\rm int},\alpha_j,a}$, $j=1,\cdots,N$, are energies corresponding to the contributions of the intersection points $\mathsf p_j$,  and depending on the angles $\alpha_j$ and the values $a$ of the applied magnetic field.
\subsection{Notation}
\begin{itemize}
	\item[]	
	\item Let $a(\kp)$ and $b(\kp)$ be two positive functions. We write 
	$a(\kp) \approx b(\kp)$ if there exist constants $\kappa_0$, $C_1$ and $C_2$ such that for all $\kp\geq\kp_0$, $C_1 a(\kp)\leq b(\kp) \leq C_2 a(\kp).$ We write $a(\kp)=\mathcal O\big(b(\kp)\big)$ if $a(\kp)/b(\kp)\rightarrow 1$ as $\kp \rightarrow +\infty$.
	\item The quantity $o(1)$ indicates a function of $\kappa$ such that $|o(1)|\rightarrow 0$ as $\kp \rightarrow +\infty$. Any expression $o(1)$ is independent of the minimizer $(\psi,\Ab)$ of~\eqref{eq:GL}.
	\item  Let  $n \in \N$ and $N \in \N$, $\gamma \in (0,1)$.  We use the following
	H\"{o}lder space
	\[C^{n,\gamma}({\overline \Om})=\left\{f \in C^n({\overline \Om})\ | \sup_{x\neq y\in \Om}\frac {|D^nf(x)-D^nf(y)|}{|x-y|^\alpha}<+\infty\right\}.\]
\end{itemize}
\subsection{Organization of the paper} Sections~\ref{sec:new_model} and~\ref{sec:eff} present the effective energies and their properties in the half-plane. Section~\ref{sec:prel} gathers useful a-priori estimates and decay results. In Section~\ref{sec:thm}, a suitable change of variables is defined and the local estimates in Theorem~\ref{thm:main} are established. 
\section{Shr\"odinger operator with a step magnetic field in the half-plane}\label{sec:new_model}
In this section, we present a Schr\"odinger operator with a step magnetic field in $\R_+^2$, introduced in~\cite{Assaad3} and whose  spectral properties are involved in the study of the effective energy in Section~\ref{sec:eff}. The spectral analysis of this operator was first done in~\cite{Assaad3}.

Let $a \in[-1,1)\setminus\{0\}$ and  $\alpha \in(0,\pi)$. We introduce the following sets in polar coordinates
\begin{align}
D_\alpha^1&=\{r(\cos\theta,\sin\theta)\in \R^2~:~r\in(0,\infty),\ 0<\theta<\alpha\}, \nonumber\\
D_\alpha^2&=\{r(\cos\theta,\sin\theta)\in \R^2~:~r\in(0,\infty),\ \alpha<\theta<\pi\}, \label{eq:A_alfa}
\end{align}
and the Schr\"odinger operator in $\R^2_+$
\begin{equation}\label{eq:P_alfa}
\mathcal H_{\alpha,a}=-\left(\nabla-i\Ab_{\alpha,a} \right)^2,
\end{equation}
where  $\Ab_{\alpha,a}=\big(0,A_{\alpha,a}\big)$ is a magnetic potential\footnote{Though one may think of a simpler choice of the magnetic potential,~\cite{Assaad3} explains the reason of defining it as in~\eqref{eq:Aa1}--\eqref{eq:Aa3}; this definition proves to be useful in explicitly deriving certain gauge results (see~\ref{lem:gauge1} and~\cite[Proof of Lemma~4.3]{Assaad3}).} defined as follows: 
\begin{equation}\label{eq:Aa1}
\mathrm{For}\ \alpha \in(0,\pi/2),\quad A_{\alpha,a}(x_1,x_2)= 
\begin{cases}
x_1+\frac {a-1}{\tan \alpha}x_2,&\mathrm{if}~(x_1,x_2)\in D^1_\alpha,\\
ax_1,&\mathrm{if}~(x_1,x_2)\in D^2_\alpha,
\end{cases}
\end{equation}
\begin{equation}\label{eq:Aa2}
\mathrm{for}\ \alpha \in(\pi/2,\pi),\quad A_{\alpha,a}(x_1,x_2)= 
\begin{cases}
x_1,&\mathrm{if}~(x_1,x_2)\in D^1_\alpha,\\
ax_1+\frac {1-a}{\tan \alpha}x_2,&\mathrm{if}~(x_1,x_2)\in D^2_\alpha,
\end{cases}
\end{equation}
\begin{equation}\label{eq:Aa3}
\mathrm{and}\quad A_{\frac \pi 2,a}(x_1,x_2)= 
\begin{cases}
x_1,&\mathrm{if}~(x_1,x_2)\in D^1_{\pi/2},\\
ax_1,&\mathrm{if}~(x_1,x_2)\in D^2_{\pi/2}.
\end{cases}
\end{equation}
The potential $\Ab_{\alpha,a}$ is in $H^1(\R^2_+;\R^2)$ and satisfies $\curl\Ab_{\alpha,a}={\mathbbm 1}_{D_\alpha^1}+a{\mathbbm 1}_{D_\alpha^2}$.
The operator $\mathcal H_{\alpha,a}$ is defined over the domain
\begin{multline*}\label{eq:domH}
\dom \mathcal H_{\alpha,a}=\big\{u\in L^2(\R^2_+)~:~ (\nabla-i\Ab_{\alpha,a})^j u \in L^2(\R^2_+),\\ \mathrm{for}\ j\in \{1,2\},(\nabla-i\Ab_{\alpha,a})\cdot (0,1)|_{\partial (\R^2_+)}=0\big\}.
\end{multline*}
Let $\mu(\alpha,a)$ be the bottom of the spectrum of $\mathcal H_{\alpha,a}$, defined by the min-max principle as follows:
\begin{equation}\label{eq:mu_alfa}
\mu(\alpha,a)=\inf_{\substack{u\in\dom \mathcal H_{\alpha,a}\\ u \neq 0 }}\frac{\|(\nabla-i\Ab_{\alpha,a})u\|_{L^2(\R^2_+)}^2}{\|u\|^2_{L^2(\R^2_+)}},
\end{equation} 
\begin{rem}\label{rem:v_0}
In~\cite[Section~3]{Assaad3}, it is asserted that $\inf \spc_{ess}(\mathcal H_{\alpha,a})=|a|\Theta_0$. It follows that if  $\mu(\alpha,a)<|a|\Theta_0$, then $\mu(\alpha,a)$ is an eigenvalue of $\mathcal H_{\alpha,a}$. This condition is crucial in deriving important properties of the effective energy in Section~\ref{sec:eff} below.
\end{rem}
%
\section{Effective energy}\label{sec:eff}
Let $a \in[-1,1)\setminus\{0\}$, $\alpha\in(0,\pi)$ and $b>0$. We introduce the following energy
\begin{equation}\label{eq:Jb}
J_{b,\alpha,a}(u)=\int_{\R_+^2} \left(b\big|(\nb-i\Ab_{\alpha,a})u\big|^2-|u|^2+\frac 12 |u|^4\right)\,dx ,
\end{equation}
where $\Ab_{a,\alpha}$ is the vector potential in Section~\ref{sec:new_model}. This energy is defined over the space
\begin{equation*}\label{eq:DJ}
H^1_{\Ab_{\alpha,a}}=\left\{u \in L^2(\R_+^2)~:~(\nb-i\Ab_{\alpha,a})u \in L^2(\R_+^2) \right\}.\end{equation*}
In what follows, we assume that $b> (|a|\Theta_0)^{-1}$.
The functional in~\eqref{eq:Jb} is bounded from below. This can be seen by using the spectral properties of the Neumann realization of the Schr\"odinger operator $(\nabla-i\Ab_{\alpha,a})^2$ on $\R_+^2$, involving those of other operators with uniform/step magnetic fields on $\R^2$ or $\R_+^2$ (see \cite[Section~2]{Assaad3}). Indeed, consider  a function $u\in C_0^\infty(\R^2)$ and let $\Gamma_\alpha:=\partial D_\alpha^1\cap \partial D_\alpha^2$.
\begin{itemize}
	\item If $\supp u\Subset D_\alpha^1\cup D_\alpha^2$, then
	\begin{equation}\label{eq:s1}
	\int_{\R_+^2} \big|(\nb-i\Ab_{\alpha,a})u\big|^2\,dx\geq |a|\int_{\R_+^2}|u|^2\,dx>|a|\Theta_0\int_{\R_+^2}|u|^2\,dx,\end{equation}
	having $\Theta_0\in (0,1)$.
	\item If $\supp u\Subset \R^2\setminus\Gamma_\alpha$ and meets the line $y=0$, then
	\begin{equation}\label{eq:s2}\int_{\R_+^2} \big|(\nb-i\Ab_{\alpha,a})u\big|^2\,dx\geq |a|\Theta_0\int_{\R_+^2}|u|^2\,dx.\end{equation}
	\item If $\supp u\Subset \R_+^2$ and meets $\Gamma_\alpha$, then
	\begin{equation}\label{eq:s3}\int_{\R_+^2} \big|(\nb-i\Ab_{\alpha,a})u\big|^2\,dx\geq \beta_a\int_{\R_+^2}|u|^2\,dx\geq |a|\Theta_0\int_{\R_+^2}|u|^2\,dx,\end{equation}
	where $\beta_a$ is the bottom of the spectrum of an operator, $\mathcal L_a$, with a step magnetic field defined over $\R^2$ in~\cite[Section~2.2]{Assaad3}.
\end{itemize}
Consequently, the lower bound of $J_{b,\alpha,a}$ is derived as follows. Let $u\in H^1_{\Ab_{\alpha,a}}$ and $R>0$. One can find a partition of unity $(\chi_k)_{k\in \N}$ of $\R^2$, satisfying 
\begin{equation*}\label{eq:partition}
\sum_k|\chi_k|^2=1, \quad \sum_k|\nabla \chi_k|^2 \leq CR^{-2},\quad {\rm and}\ \supp \chi_k\subset B_k(R),
\end{equation*}
where $(B_k)_{k}$ are balls of radii $R$ such that $B_1(R)=B(O,R)$. 
 Using the IMS localization formula (see~\cite[Theorem~3.2]{cycon2009schrodinger}), we have
\begin{align}
J_{b,\alpha,a}(u)&\geq \sum_k J_{b,\alpha,a}(\chi_k u)-b\sum_k \|\nabla\chi_k u\|^2_{L^2(\R_+^2)}\nonumber \\
& \geq\sum_{k \neq 1} \Big( J_{b,\alpha,a}(\chi_k u)-\frac {Cb}{R^2} \|\chi_k u\|^2_{L^2(\R_+^2)}\Big)+J_{b,\alpha,a}(\chi_1 u)-\frac {Cb}{R^2} \|\chi_1u\|^2_{L^2(B_1(R))}\nonumber \\
&\geq (b|a|\Theta_0-1-\frac {Cb}{R^2} )\sum_{k\neq 1}\|\chi_k u\|^2_{L^2(\R_+^2)}+\frac 12\big \|\big(|\chi_1 u|^2-(1+\frac {Cb}{R^2})\big)\big\|^2_{L^2(B_1(R))}\nonumber \\
&\quad -\frac 12 (1+\frac {Cb}{R^2})^2|B_1(R)|,\qquad \mbox{(by~\eqref{eq:s1}--\eqref{eq:s3})} \nonumber \\
& \geq (b|a|\Theta_0-1-\frac {Cb}{R^2} )\sum_{k\neq 1}\|\chi_k u\|^2_{L^2(\R_+^2)}-\frac {CR^2}2 (1+\frac {Cb}{R^2})^2.\nonumber
\end{align}
Having $b> (|a|\Theta_0)^{-1}$, we choose  $R$ sufficiently large so that $b|a|\Theta_0-1-Cb/R^2 >0$. Hence, we get the desired lower bound.

We  define now the (finite) ground-state energy
\begin{equation}\label{eq:gs-E}
E_{\alpha,a}(b)=\displaystyle\inf_{u \in H^1_{\Ab_{\alpha,a}}}
J_{b,\alpha,a}(u). 
\end{equation} 
Working under the assumption 
\begin{equation}\label{ass_mu}
\mu(\alpha,a)<|a|\Theta_0,
\end{equation}
where $\mu(\alpha,a)$ is the value in~\eqref{eq:mu_alfa}, important properties of the functional in~\eqref{eq:Jb} are the following (cf.~\cite[Proposition~15.3.10]{fournais2010spectral} for the same argument). This functional is non-positive. It has non-trivial minimizers, $u_{b,\alpha,a}$, if and only if \footnote{Note that the condition in~\eqref{eq:cond} is only valid under~\eqref{ass_mu}.  In light of Remark~\ref{rem:v_0}, the eigenfunction corresponding to $\mu(\alpha,a)$ is used, in such situations, to construct a test function proving the non-triviality of the minimizers.
	If~\eqref{ass_mu} is not satisfied, then $\mu(\alpha,a)=|a|\Theta_0$. In this case, for $b>(|a|\Theta_0)^{-1}$, $J_{b,\alpha,a}$ will only have zero minimizers, and the study will be trivial.\label{foot}}
\begin{equation}\label{eq:cond}
\frac 1{|a|\Theta_0}< b<\frac 1{\mu(\alpha,a)}.\end{equation}
These minimizers satisfy 
\begin{equation*}\label{eq:u-alfa}
\|u_{b,\alpha,a}\|_{L^\infty(\R_+^2)}\leq1,
\end{equation*}
and
\begin{equation}\label{eq:u-alfa1}
\int_{\R_+^2}e^{2\delta|x|}\left(|u_{b,\alpha,a}(x)|^2+|(\nabla-i\Ab_{\alpha,a})u_{b,\alpha,a}(x)|^2\right)\,dx\leq C,
\end{equation}
where $\delta,C$ are two positive constants dependent of $b$, $\alpha$ and $a$. 

In light of the discussion above, we have
\begin{align}\label{eq:cond1}
E_{\alpha,a}(b)&<0,\qquad \mbox{for}\ \frac 1{|a|\Theta_0}< b<\frac 1{\mu(\alpha,a)}, \nonumber\\
E_{\alpha,a}(b)&=0,\qquad \mbox{for}\ b\geq\frac 1{\mu(\alpha,a)}.
\end{align}
\begin{remark}
	Using a symmetry argument in the particular case where $\alpha=\pi/2$ and $a=-1$, $E_{\pi/2,-1}(b)$ will be the same effective energy introduced in~\cite{bonnaillie2007superconductivity} (up to a scaling factor).
	\end{remark}
\section{Preliminaries}\label{sec:prel}
In this section, we present some properties of the minimizers $(\psi, \Ab)$ of $\GL$ in~\eqref{eq:GL}, needed in the derivation our main results.
\subsection{Classical Estimates}\label{sec:a-p-est}
Recall the magnetic field $B_0$ introduced in Assumption~\ref{assump1}.  We fix a magnetic potential $\Fb\in \Hd$ generating $B_0$  (see~\cite[Lemma A.1]{Assaad}):
\begin{equation}\label{eq:Fb}
\exists\,{\rm unique}\ \Fb\in \Hd\ \mbox{such that}\ \curl\Fb=B_0.\end{equation}
The theorem below introduces some  estimates of the critical points $(\psi, \Ab)$ of $\GL$ in~\eqref{eq:GL}, involving the field $\Fb$. In light of this theorem, $\Fb$ proves to be a useful approximation of any vector potential $\Ab$ corresponding to a minimizer of~\eqref{eq:GL}.
\begin{thm}\label{thm:priori}
	Let $\gamma\in(0,1)$ be a constant. 
	Under Assumption~\ref{assump1}, there exists a constant $C>0$  such that if
	$(\psi,\Ab)\in H^1(\Om;\C)\times \Hd$ is a weak solution of~\eqref{eq:Euler}, then 
	\begin{enumerate}
		\item $\|\psi\|_{L^\infty(\Omega)}\leq 1$,
		\item $\|(\nb-i \kp H\Ab)\psi\|_{L^2(\Om)}\leq C\kp\|\psi\|_{L^2(\Om)}$,
		\item $\displaystyle\|{\curl(\Ab-\Fb)}\|_{L^2(\Om)}\leq \frac C
		H\|\psi\|^2_{L^2(\Om)}$,
		\end{enumerate}
		and, under the additional Assumption~\ref{A_2}
	\begin{enumerate}
		\setcounter{enumi}{3}
		\item $\Ab-\Fb\in H^2(\Omega)$ and $\displaystyle\|\Ab-\Fb\|_{H^2(\Om)}\leq \frac C
		\kappa$,
		\item $\Ab-\Fb\in C^{0,\gamma}(\overline{\Om})$ and $\displaystyle\|\Ab-\Fb\|_{C^{0,\gamma}(\overline \Om)}\leq \frac C \kappa$.
	\end{enumerate}
\end{thm}
The proof of this theorem can be found in~\cite[Theorem~4.2]{Assaad} and~\cite[Proposition~10.3.1 \& Lemma~10.3.2]{fournais2010spectral}.
\subsection{Exponential decay of the order parameter}
The following theorem displays regimes for the intensity of the applied magnetic field where the order parameter and the GL energy are exponentially small away from certain intersection points of the magnetic edge and the boundary.

We suppose that the assumptions in Section~\ref{sec:setting} are satisfied. Let $j\in \{1,\cdots,n\}$. We define the set
\[S=\left\{\mathsf p_j\in\Gamma\cap\partial\Om~:~b<\big(\mu(\alpha_j,a)\big)^{-1}\right\}.\]
\begin{theorem}\label{thm:decay}
	Assume that $b$ satisfy
		\[(|a|\Theta_0)^{-1}< b<\big(\min_{j\in\{1,\cdots,n\}}\mu(\alpha_j,a)\big)^{-1}.\]
There exist constants $\kappa_0>0$, $C>0$, and $\delta_0>0$  such that, if
	\[\kappa \geq \kappa_0,\ \kappa_0 \kappa^{-1}\leq \ell<1 ,\ \mathrm{and}\ (\psi,\Ab)~\mathrm{is~a~solution~of}~\eqref{eq:Euler},\] 
	then
	\begin{equation*}
	\int_{\Omega\cap\{\mathrm{dist}(x,S)\geq \ell\}}
	\Big(|\psi|^2+ (\kappa H)^{-1}|(\nabla-i\kappa H\Ab)\psi|^2\Big)\,dx \leq C  \kappa^{-1} e^{-\delta_0\kappa\ell}.
	\end{equation*}
\end{theorem}

\begin{proof}[Proof of Theorem~\ref{thm:decay}]
	The proof is a consequence of  the  decay estimates established in~\cite[Theorem~1.6]{Assaad3}; indeed, for $b\in \Big((|a|\Theta_0)^{-1},\big(\min_{j\in\{1,\cdots,n\}}\mu(\alpha_j,a)\big)^{-1}\Big)$,  there exist $\kappa_0,C,\beta>0$ such that, for $\kappa\geq\kappa_0$ and $H=b\kappa$,
	\begin{multline}\label{eq:exp_decay}
	\int_{\Omega\cap\{{\rm dist}(x,S)\geq \frac{1}{\sqrt{\kappa H}}\}}e^{\beta\sqrt{\kappa H}\,{\rm dist}(x,S)}
	\Big(|\psi|^2+{\frac{1}{\kp H}}|(\nabla-i\kappa H\Ab)\psi|^2\Big)\,dx\\
	\leq
	C \int_{\Omega\cap\{{\rm dist}(x,S)\leq \frac{1}{\sqrt{\kappa H}}\}.}
	|\psi|^2\,dx.\end{multline}
	 We choose $\kappa_0$ so that $\kappa_0\geq 1/\sqrt{b}$. Thus, for  $\kappa\geq\kappa_0$ and $\kappa_0 \kappa^{-1}\leq \ell<1$, we get $\ell\geq 1/\sqrt{\kappa H}$. 
	Using~\eqref{eq:exp_decay},  one can easily verify the claim of Theorem~\ref{thm:decay}, with $\delta_0=\delta_0(b)=\beta\sqrt{b}$.
\end{proof}

\section{Local estimates of minimizers (Proof of Theorem~\ref{thm:main})}\label{sec:thm}
We are still working under the assumptions in Section~\ref{sec:setting}. The aim of this section is to study the concentration of the minimizers $(\psi,\Ab)$ of the functional~\eqref{eq:GL} near the intersection points, $\mathsf p_j$, of $\partial \Om$ and $\Gamma$. This will be displayed by local estimates of the GL ground-state energy and the $L^4$-norm of  minimizers, that establish Theorem~\ref{thm:main}.

\subsection{Change of variables}\label{sec:Psi}
We will carry out the computation in adapted coordinates near  $\partial \Om \cap\Gamma$. The properties in this section are borrowed from~\cite[Section~4]{Assaad3}, where these coordinates are first defined. We present them below for the convenience of the reader, and  we refer to the aforementioned paper for more details.

 For $j  \in\{1,\cdots,n\}$, there exist $r_j>0$ and a local diffeomorphism $\Psi=\Psi_j$ of $\R^2$ satisfying the following:  
\begin{equation*}\label{eq:Psi2}
\Psi(\mathsf p_j)=(0,0)\,,\qquad |J_\Psi|(\mathsf p_j)=|J_{\Psi^{-1}}|(0,0)=1,
\end{equation*}
and there exists a neighbourhood $\mathcal U_j$ of $(0,0)$ such that 

\[\Psi\big(B(\mathsf p_j,r_j)\cap \Om_1\big) =\mathcal U_j \cap  D^{\alpha_j}_1\,,\quad  \Psi\big(B(\mathsf p_j,r_j)\cap \Om_2\big) =\mathcal U_j \cap  D^{\alpha_j}_2,\]
\[\Psi\big(B(\mathsf p_j,r_j)\cap (\partial \Om_1 \setminus \Gamma)\big) =\mathcal U_j \cap \R_+\times\{0\},\]
\[\Psi\big(B(\mathsf p_j,r_j)\cap (\partial \Om_2 \setminus \Gamma)\big) =\mathcal U_j \cap \R_-\times\{0\},\]
\[\Psi\big(B(\mathsf p_j,r_j)\cap \Gamma \big) =\mathcal U_j \cap (\hat x_2=\hat x_1\tan \alpha_j).\]
Here, $(\hat x_1,\hat x_2):=\Psi(x_1,x_2)$, and the sets $D^{\alpha_j}_1$ and $D^{\alpha_j}_2$ are defined in~\eqref{eq:A_alfa}.
We assume further that the radii $r_j$ are sufficiently small so that $\big(B(\mathsf p_j,r_j)\big)_{j\in\{1,\cdots,n\}}$ is a family of disjoint balls.	Using the properties above, one can prove the existence of a constant $C>0$, \emph{independent of $j$}, such that the Jacobians $J_\Psi$ and $J_{\Psi^{-1}}$ satisfy
\begin{equation}\label{eq:Psi3}
\big||J_\Psi(x)|-1 \big|\leq C \ell\qquad \mathrm{and}\qquad \big||J_{\Psi^{-1}}(\hat x)|-1 \big|\leq C \ell,
\end{equation}
for all $x \in B(\mathsf p_j,\ell) \subset B(\mathsf p_j,r_j)$ and $\hat x=\Psi(x)$. 
Let $\Eb=(E_1,E_2)\in H^1(\Om;\R^2)$ be such that $\curl \Eb=B$, for $B \in L^2(\R^2)$, 
and  $u \in H^1(\Om;\C)$ 
 such that $\supp u \subset B(\mathsf p_j,r_j)$. Consider the magnetic potential $\hat {\Eb}=(\hat{E}_1,\hat{E}_2) \in H^1\big(\Psi\big(B(\mathsf p_j,r_j)\big)\cap \R^2_+;\R^2\big)$  satisfying $\hat{E}_1\,d\hat{x}_1+\hat{E}_2\,d\hat{x}_2=E_1\,d x_1+E_2\,dx_2$, and the function $\hat{u}$, defined in $\Psi(B(\mathsf p_j,r_j))\cap \R^2_+$ by $\hat{u}(\hat{x})=u\big(\Psi^{-1}(\hat{x})\big)$.  Furthermore, let 
\[\hat B(\hat{x})=B\big(\Psi^{-1}(\hat{x})\big),\qquad \mbox{for all}\ \hat x \in \Psi(B(\mathsf p_j,r_j))\cap \R^2_+.\] 
One can check that 
\begin{equation*}\label{eq:curl-hat}
\curl \hat{\Eb}=\partial_{\hat{x}_1}\hat{E}_2-\partial_{\hat{x}_2}\hat{E}_1=\hat{B}J_{\Psi^{-1}},
\end{equation*}
and, for any $\mathfrak b>0$
\begin{multline}
\int_{\Om}\big|(\nabla-i\mathfrak b\Eb)u\big|^2\,dx
=\\\int_D\sum_{1\leq k,m\leq 2} G_{k,m}(\hat x)\big(\partial_{\hat x_k}-i\mathfrak b\hat E_k\big)\hat {u}(\hat {x})\overline{\big(\partial_{\hat x_m}-i\mathfrak b\hat E_m\big)\hat {u}(\hat {x})}\,|J_{\Psi^{-1}}(\hat {x})|\,d\hat {x}.\label{eq:Q_trans}
\end{multline}
Here 
$D=\Psi(B(\mathsf p_j,r_j))\cap \R^2_+$ and $G_{k,m}(\hat x)$ are the elements of the matrix $G(\hat x)=(d \Psi)(d \Psi)^t\,_{|\Psi^{-1}(\hat x)}$.
For any $\ell<r_j$, we have 
\begin{equation}\label{eq:Bon1}|G_{k,m}(\hat x)-\delta_{k,m}|\leq C\ell,\quad\ \hat x\in \Psi(B(\mathsf p_j,\ell)\big)\end{equation}
for some $C>0$ independent of $j$. 

Using the coordinates transformation above, the field $\Fb$ in~\eqref{eq:Fb} can be expressed in the following canonical manner:
\begin{lemma}\label{lem:gauge1}
	Let $a\in [-1,1)\setminus \{0\}$, and $B(0,l)\subset \Psi\big(B(\mathsf p_j,r_j)\big)$ be a ball of radius $l$.  Consider the  vector potential $\Fb\in \Hd$ satisfying $\curl \Fb=\mathbbm 1_{\Om_1}+a\mathbbm 1_{\Om_2}$.  There exists a function $\varphi_{j,l}\in H^2\big(B(0,l)\cap \R^2_+\big)$ such that the vector potential $\hat \Fb_{\rm g}:=\hat \Fb-\nabla_{\hat{x}_1,\hat{x}_2} \varphi_{j,l}$, defined in $B(0,l)\cap \R^2_+$, satisfies
	\begin{equation*}
	\big(\hat{F}_{\rm g}\big)_1=0,\quad
	\big(\hat{F}_{\rm g}\big)_2= A_{\alpha,a}+f,
	\end{equation*}
	where $A_{\alpha,a}$ is the potential introduced in~\eqref{eq:P_alfa}, $f$ is a continuous function satisfying $|f(\hat{x}_1,\hat{x}_2)|\leq C (\hat{x}^2_1+|\hat x_1\hat x_2|)$, for some $C>0$ independent of $j$.
\end{lemma}
\subsection{A useful lower bound}
We introduce the following local energies of any configuration $(\psi,\Ab)\in H^1(\Om;\C)\times \Hd$, in a domain $D\subset \Om$
\begin{align}
\mathcal E_0(\psi,\mathbf{A};D) &=\int_D\big( |(\nb-i \kp H \mathbf{A})\psi|^2-\kp^2|\psi|^2 +\frac 12 \kp^2|\psi|^4 \big)\,dx,\label{eq:local1}\\
\mathcal E(\psi,\mathbf{A};D)&=\mathcal E_0(\psi,\mathbf{A};D) +(\kp H)^2 \int_\Omega|\curl(\mathbf{A}-\mathbf{F})|^2\,dx.\label{eq:local2}
\end{align}

Let $\ell\in(0,1)$. For each $\mathsf p_j\in\Gamma\cap \partial \Om$, $j\in\{1,\cdots,n\}$, recall the set $\mathcal N_j(\ell)$ defined in~\eqref{eq:Nj} by  
\begin{equation*}
\mathcal N_j(\ell)=\{x\in \Om,\ \dist(x,\mathsf p_j)\leq\ell\}.
\end{equation*} 
In Proposition~\ref{prop:E0low} below, we establish a lower bound for the local energy of an arbitrary function $u\in H^1(\Om;\C)$ supported in a neighbourhood of $\mathsf p_j$, which will be helpful in deriving the local estimates in Theorem~\ref{thm:main}.
\begin{proposition}\label{prop:E0low}
	There exist two constants $\kappa_0>1$ and  $C>0$ such that, for $\kappa \geq \kappa_0$ and for any $\mathsf p_j\in\Gamma\cap \partial \Om$, the following is true. If 
	\begin{itemize}
		\item $(\psi,\Ab)\in H^1(\Omega;\C)\times \Hd$ is a solution of~\eqref{eq:Euler}.
		\item $u \in H^1(\Omega;\C)$ such that $\supp u\subset B(\mathsf p_j,\ell)$ and $|u| \leq 1$.
	\end{itemize}  
	then 
	\begin{equation*}
	\mathcal E_0\big(u,\Ab;\mathcal N_{j}(\ell)\big) \geq b^{-1} E_{\alpha_j,a}(b) -C(\kappa^{\frac 32}\ell^2+\kappa^{\frac 52}\ell^{\frac {10}3}+\kappa^2\ell^3+\kappa^{\frac 92}\ell^6),
	\end{equation*} 
	where $\mathcal E_0$ is the functional in~\eqref{eq:local1}, and $E_{\alpha_j,a}(b)$ is the energy in~\eqref{eq:gs-E}.
\end{proposition}
\begin{proof}
	Let $\gamma \in (0,1)$ and $\Fb$ be the vector field introduced in~\eqref{eq:Fb}. We define the function $\phi_{j}$ by
	\begin{equation}\label{eq:phi-j}
	\phi_{j}(x)=\Big(\Ab(\mathsf p_j)-\Fb(\mathsf p_j)\Big)\cdot x.
	\end{equation} As a consequence of the fifth item in Theorem~\ref{thm:priori}, we get the following approximation of the vector potential  $\Ab$
	\begin{equation}\label{eq:AF}
	|\Ab(x)-\nabla\phi_{j}(x)-\Fb(x)| \leq \frac C\kappa \ell^\gamma\,,\qquad \mathrm{for}\ x \in \mathcal N_{j}(\ell).
	\end{equation} 
	We choose $\gamma=2/3$ in~\eqref{eq:AF}. Let $v=e^{-i\kappa H\phi_{j}}u$. Using~\eqref{eq:AF},  Cauchy's inequality, and the bound $|v|\leq 1$, we may write
	\begin{equation}\label{E_0}
	\mathcal E_0(u,\Ab;\mathcal N_{j}(\ell)) \geq (1-\kappa^{-\frac 12})\mathcal E_0(v,\Fb;\mathcal N_{j}(\ell))-C\Big(\kappa^{\frac 32}\ell^2+\kappa^{\frac 52}\ell^{\frac {10}3}\Big).
	\end{equation}
	Now, we use the change of variables introduced in Section~\ref{sec:Psi}, valid in a neighbourhood of $\mathsf p_j$, to locally send  the domain in $\Om$ onto $\R^2_+$.  $\kappa$ is assumed sufficiently large so that
	$B(\mathsf p_j,\ell) \subset B(\mathsf p_j,r_j)$.
	We associate to  $v$ the function $\hat{v}=v\circ\Psi^{-1}$, defined in $\Psi\big( B(\mathsf p_j,\ell)\big)$.
	We may use the transformation formula in~\eqref{eq:Q_trans} and the  properties in~\eqref{eq:Psi3} and~\eqref{eq:Bon1} to conclude that
	\begin{multline}\label{eq:Q_hat}
	(1-C\ell) \int_{\Psi( B(\mathsf p_j,\ell))\cap \R^2_+}\big|(\nabla-i\kappa H\hat{\Fb})\hat{v}\big|^2\,d\hat x\leq \int_{\Om}\big|(\nabla-i\kappa H\Fb)v\big|^2\,dx\\\leq (1+C\ell)\int_{\Psi( B(\mathsf p_j,\ell))\cap \R^2_+}\big|(\nabla-i\kappa H\hat{\Fb})\hat{v}\big|^2\,d\hat x,
	\end{multline}
	where $\hat \Fb$ is the transform of $\Fb$ by $\Psi$, and $C>0$ is a constant independent of $j$.
	In addition, due to the support of $v$ and~\eqref{eq:Psi3}, we note the existence of $c_1>0$  such that 
	$\Psi \big(B(\mathsf p_j,\ell)\big)\subset B(0,c_1 \ell) \subset \Psi \big(B(\mathsf p_j,r_j)\big)$, for large $\kappa$. 
	Consequently, the gauge transform in Lemma~\ref{lem:gauge1} allows us to write
	\begin{multline}\label{eq:gauge2}
	\int_{\Psi( B(\mathsf p_j,\ell))\cap \R^2_+}\big|(\nabla-i\kappa H\hat{\Fb})\hat{v}\big|^2\,d\hat x\\=\int_{\Psi( B(\mathsf p_j,\ell))\cap \R^2_+} \big|(\nabla-i\kappa H\hat{\Fb}_{\rm g})\hat{v}_{\rm g}\big|^2\,d\hat{x},
	\end{multline}
	where
	$\hat{v}_{\rm g}(\hat x)=
	\hat{v}(\hat x)e^{-i\kappa H\varphi(\hat x)}$, for $\hat x\in \Psi\big( B(\mathsf p_j,\ell)\big) \cap \R^2_+$.
	Here $\varphi=\varphi_{j,l}$, for $l=c_1\ell$, is the gauge function in Lemma~\ref{lem:gauge1}, and $\hat{\Fb}_{\rm g}$ is  the magnetic potential in the aforementioned lemma. 
	
	Recall the potential $\Ab_{\alpha,a}$ introduced in~\eqref{eq:P_alfa}. Extending $\hat{v}$ and $\hat{v}_{\rm g}$ by zero in $\R^2_+$, the Cauchy's inequality applied in~\eqref{eq:gauge2}, and the support of the function $\hat{v}_{\rm g}$ imply
	\begin{multline}\label{eq:gauge3}
	\int_{\Psi( B(\mathsf p_j,\ell))\cap \R^2_+}\big|(\nabla-i\kappa H\hat{\Fb}_{\rm g})\hat{v}_{\rm g}\big|^2\,d\hat x\geq (1-\kappa^{-\frac 12})\int_{\R^2_+} \big|(\nabla-i\kappa H\Ab_{\alpha_j,a})\hat{v}_{\rm g}\big|^2\,d\hat{x}\\-C\kappa^\frac 92\ell^4\int_{\R^2_+}|\hat{v}_{\rm g}\big|^2\,d\hat{x},
	\end{multline}
	where  $\alpha_j$ is the corresponding angle to the point $\mathsf p_j$, defined in Notation~\ref{not:alfa}. 
	But 
	\begin{equation*}
	\int_{\R^2_+}|\hat{v}_{\rm g}\big|^2\,d\hat{x}=\int_{B(\mathsf p_j,\ell)\cap\Omega} |v|^2\,|J_\Psi|\,dx.\end{equation*}
	Thus, using~\eqref{eq:Psi3} we get
	\begin{equation}\label{eq:v_v_hat} 
	(1-C\ell)\int_\Omega |v|^2\,dx\leq \int_{\R^2_+}|\hat{v}_{\rm g}\big|^2\,d\hat{x}\leq (1+C\ell)\int_\Omega |v|^2\,dx.
	\end{equation}
	Plug~\eqref{eq:v_v_hat} into~\eqref{eq:gauge3}, and use again $|v|\leq1$ together with its support to obtain
	\begin{equation}\label{eq:gauge5}
\int_{\Psi( B(\mathsf p_j,\ell))\cap \R^2_+}\big|(\nabla-i\kappa H\hat{\Fb}_{\rm g})\hat{v}_{\rm g}\big|^2\,d\hat x\geq (1-\kappa^{-\frac 12})\int_{\R^2_+} \big|(\nabla-i\kappa H\Ab_{\alpha_j,a})\hat{v}_{\rm g}\big|^2\,d\hat{x}-C\kappa^\frac 92\ell^6.
	\end{equation}
	Similarly to~\eqref{eq:v_v_hat}, we have
	\begin{equation}\label{eq:v_v_hat1}
	(1-C\ell)\int_\Omega |v|^4\,dx\leq \int_{\R^2_+}|\hat{v}_{\rm g}\big|^4\,d\hat{x}\leq (1+C\ell)\int_\Omega |v|^4\,dx.
	\end{equation}	
Consequently, using~\eqref{eq:Q_hat},~\eqref{eq:gauge2} and~\eqref{eq:v_v_hat}--\eqref{eq:v_v_hat1}, we retake the energy $\mathcal E_0(v,\Fb;\mathcal N_{j}(\ell))$ appearing  in~\eqref{E_0} and write
\begin{align}\label{E_1}
&\mathcal E_0(v,\Fb;\mathcal N_{j}(\ell))= \int_{\Om}\Big(\big|(\nabla-i\kappa H\Fb)v\big|^2-\kappa^2|v|^2+\frac {\kappa^2}2  |v|^4\Big)\,dx\nonumber\\
&\geq (1-C\ell-\kappa^{-\frac 12})\left(\int_{\R_+^2}\Big(\big|(\nabla-i\kappa H\Ab_{\alpha_j,a})\hat v_{\rm g}\big|^2 -\kappa^2 |\hat v_{\rm g}|^2+\frac {\kappa^2}2  |\hat v_{\rm g}|^4\Big)\,d\hat x\right)
-r(\kappa),
\end{align}
where $r(\kappa)=C\big(\kappa^{\frac 32}\ell^2+\kappa^{\frac 92}\ell^6+\kappa^2\ell^3\big)$.
	Next, we use the scaling $t=\sqrt{\kappa H}\hat x=\sqrt{b}\kappa \hat x$ (see Assumption~\ref{A_2}), and  define 
	\begin{equation*}\label{eq:v-scal}
	{\rm v}(t)=\hat v_{\rm g}\Big(\frac t{\sqrt{b}\kappa}\Big),\quad {\rm for}\ t\in \R_+^2.
	\end{equation*}
	One can simply check that
	\begin{equation*}
	\int_{\R_+^2}\big|(\nabla_{\hat x}-i\kappa H\Ab_{\alpha_j,a}(\hat x))\hat v_{\rm g}(\hat x)\big|^2\,d\hat x=\int_{\R_+^2}\big|(\nabla_{t}-i\Ab_{\alpha_j,a}(t)){\rm v}(t)\big|^2\,dt 
	\end{equation*}
	and
	\begin{equation*}
	\int_{\R_+^2}|\hat v_{\rm g}(\hat x)|^2\,d\hat x=\frac 1{b\kappa^2}\int_{\R_+^2}|{\rm v}(t)|^2\,dt,\quad \int_{\R_+^2}|\hat v_{\rm g}(\hat x)|^4\,d\hat x=\frac 1{b\kappa^2}\int_{\R_+^2}|{\rm v}(t)|^4\,dt.
	\end{equation*}
	Hence, 
	\begin{equation}\label{E_2}
	\int_{\R_+^2}\Big(\big|(\nabla-i\kappa H\Ab_{\alpha_j,a})\hat v_{\rm g}\big|^2 -\kappa^2 |\hat v_{\rm g}|^2+\frac {\kappa^2}2  |\hat v_{\rm g}|^4\Big)\,d\hat x=\frac 1 b J_{b,\alpha_j,a}({\rm v})\geq \frac 1b E_{\alpha_j,a}(b),
	\end{equation}
	where $J_{b,\alpha_j,a}$ and $E_{\alpha_j,a}(b)$ are  the energies in~\eqref{eq:Jb} and~\eqref{eq:gs-E} respectively.
	Having $E_{\alpha_j,a}(b)\leq 0$, we put~\eqref{E_2} in~\eqref{E_1} and get
	\begin{equation}\label{E_3}
	\mathcal E_0(v,\Fb;\mathcal N_{j}(\ell))\geq \frac 1b E_{\alpha_j,a}(b)-C\kappa^{\frac 32}\ell^2-C\kappa^{\frac 92}\ell^6-C\kappa^2\ell^3.
	\end{equation}
	Implement~\eqref{E_3} in~\eqref{E_0} to complete the proof.	 
\end{proof}
\subsection{Proof of the main result}
\begin{proof}[Proof of Theorem~\ref{thm:main}]
	Let $\gamma \in (0,1)$ and $\hat \ell=(1+\gamma)\ell$. We assume that $\kappa$ is sufficiently large so that $\mathcal N_j(\hat\ell)\cap \mathcal N_k(\hat\ell)=\emptyset$, for any $j,k\in\{1,\cdots,n\}$, $j\neq k$.
	
\emph{Let $j\in T$}. Consider a smooth function $f_j$ satisfying
	\begin{equation}\label{eq:f}
	f_j=1\ \mathrm{in}\ \mathcal N_j(\ell),\ f_j=0 \ \mathrm{in}\  \mathcal N_j\big(\hat\ell\big)^\complement,
	0\leq f_j \leq 1\ \mathrm{and}\  |\nabla f_j| \leq C\gamma^{-1} \ell^{-1}\ \mathrm{in}\ \Omega.
	\end{equation}
	Here, $ \mathcal N_j\big(\hat\ell\big)^\complement$ denotes the complement of $ \mathcal N_j\big(\hat\ell\big)$ in $\Om$. We have the following  identity
	\begin{multline}\label{eq:Delta}
	\int_{\mathcal N_j(\hat{\ell})}\big|(\nb-i \kp H \Ab)f_j\psi\big|^2\,dx =\int_{\mathcal N_j(\hat{\ell})}\big|f_j(\nb-i \kp H \Ab)\psi\big|^2\,dx +\int_{\mathcal N_j(\hat{\ell})}|\nabla f_j|^2\psi|^2\,dx\\+2\re \int_{\mathcal N_j(\hat{\ell})}f_j(\nb-i \kp H \Ab)\psi\cdot\psi\nabla f_j\,dx.
	\end{multline}
	Consider the following obvious decompositions
	\begin{align}\label{eq:f2}
	\int_{\mathcal N_j(\hat{\ell})}f_j^2|\psi|^2\,dx&=\int_{\mathcal N_j(\hat{\ell})}|\psi|^2\,dx+\int_{\mathcal N_j(\hat{\ell})}(f_j^2-1)|\psi|^2\,dx\nonumber \\
	\int_{\mathcal N_j(\hat{\ell})}f_j^4|\psi|^4\,dx&=\int_{\mathcal N_j(\hat{\ell})}|\psi|^4\,dx+\int_{\mathcal N_j(\hat{\ell})}(f_j^4-1)|\psi|^4\,dx,
	\end{align}
	and that of $ \int_{\mathcal N_j(\hat{\ell})}\big|f_j(\nb-i \kp H \Ab)\psi\big|^2\,dx$ into
	\[\int_{\mathcal N_j(\hat{\ell})}\big|(\nb-i \kp H \Ab)\psi\big|^2\,dx+\int_{\mathcal N_j(\hat{\ell})}(f_j^2-1)\big|(\nb-i \kp H \Ab)\psi\big|^2\,dx.\] Moreover, note that 
	\begin{equation}\label{eq:f1}
	\Big|\re\int_{\mathcal N_j(\hat{\ell})}f_j(\nb-i \kp H \Ab)\psi\cdot\psi\nabla f_j\,dx\Big|\leq \|f_j(\nb-i \kp H \Ab)\psi\|_{L^2(\mathcal N_j(\hat{\ell}))}\|\psi\nabla f_j\|_{L^2(\mathcal N_j(\hat{\ell})),}\end{equation}
	and recall that $\ell\approx\kappa^{-\rho}$ is chosen so that $\ell\gg\kappa^{-1}$. Hence, using~\eqref{eq:Delta}--\eqref{eq:f1} and the properties of $f_j$ in~\eqref{eq:f}, particularly that $\nabla f_j$, $f_j^2-1$ and $f_j^4-1$ are supported in $\mathcal  N_j({\ell})^\complement$, together with  Theorem~\ref{thm:decay} which ensures that $|\psi|$ is exponentially small in $\mathcal  N_j(\hat{\ell})\setminus \mathcal N_j({\ell})$, we get
	\begin{equation}\label{eq:f3}
\mathcal E_0(f_j\psi,\Ab;\mathcal N_j(\hat{\ell}))=\mathcal E_0(\psi,\Ab;\mathcal N_j(\hat{\ell}))+o(1).
	\end{equation}
	\paragraph{\emph {A lower bound of the local energy}}
	Notice that the function $f_j\psi$ satisfies the conditions in Proposition~\ref{prop:E0low}, with $\hat\ell$ replacing $\ell$, 
	then using this proposition and~\eqref{eq:f3}, we get (with the choice of $\rho\in(4/5,1)$ in $\ell\approx\kappa^{-\rho}$)
	\begin{equation}\label{eq:f4}
	\mathcal E_0(\psi,\Ab;\mathcal N_j(\hat{\ell}))\geq\frac 1 b E_{\alpha_j,a}(b)+o(1).
	\end{equation}
	\paragraph{\emph {An upper bound of the local energy}}
	Inspired by~\cite{sandier2003decrease,helffer2017decay,helffer2018density,Assaad2019}, we define the following test function
	\begin{equation*}\label{eq:test}
	w(x)=\mathbbm{1}_{\mathcal N_j(\hat\ell)}(x) f_j(x)e^{i\kappa H(\phi_j(x)+\varphi_{j,\hat \ell}(\Psi_j(x))}u_j(\sqrt{\kappa H}\Psi_j(x))+(1-f_j(x))\psi(x)\,,
	\end{equation*}	
	where $\Psi_j$ is the coordinate transformation in Section~\ref{sec:Psi}, $f_j$, $\phi_j$, and $\varphi_{j,\hat \ell}$ are respectively the functions in~\eqref{eq:f},~\eqref{eq:phi-j}, and Lemma~\ref{lem:gauge1},  and $u_j=u_{b,\alpha_j,a}$ is a minimizer of the functional in~\eqref{eq:Jb} with $b=H/\kappa$.
	
	A minimizer $(\psi,\Ab)$ of~\eqref{eq:GL} obviously satisfies \[\mathcal E_{\kp,H}(\psi,\Ab)\leq \mathcal E_{\kp,H}(w,\Ab).\] Suppressing the term $\kp^2H^2\|\curl(\Ab-\Fb)\|_{L^2(\Om)}^2$  from the above expression, we get
\[\mathcal E_0(\psi,\Ab;\Om)\leq \mathcal E_0(w,\Ab;\Om).\]	
Note that $f_j=0$ in ${\mathcal N_j(\hat\ell)}^\complement$, hence using the following decompositions
\begin{align*}
\mathcal E_0(\psi,\Ab;\Om)&= \mathcal E_0(\psi,\Ab;\mathcal N_j(\hat\ell))+\mathcal E_0(\psi,\Ab;{\mathcal N_j(\hat\ell)}^\complement)\\
\mathcal E_0(w,\Ab;\Om)&= \mathcal E_0(w,\Ab;\mathcal N_j(\hat\ell))+\mathcal E_0(\psi,\Ab;{\mathcal N_j(\hat\ell)}^\complement)
\end{align*}
we further get
\begin{equation}\label{eq:f5}
\mathcal E_0(\psi,\Ab;\mathcal N_j(\hat\ell))\leq\mathcal E_0(w,\Ab;\mathcal N_j(\hat\ell)).
\end{equation}	
On the other hand, the decay estimates in~\eqref{eq:u-alfa1} and Theorem~\ref{thm:decay} assure that $u_j(\sqrt{\kappa H}\Psi_j(x))$ and $|\psi(x)|$ are exponentially small in  $\mathcal N_j(\hat\ell)\setminus\mathcal N_j(\ell)$ (having $\ell\gg\kappa^{-1}$). Hence, a  computation\footnote{we omit the computation details which  is now straightforward, after having the proofs in the above paragraph and Proposition~\ref{prop:E0low}.} quite similar to the one done  in the lower bound proof above and in Proposition~\ref{prop:E0low} yields
\begin{equation}\label{eq:f6}
\mathcal E_0(w,\Ab;\mathcal N_j(\hat\ell))=\frac 1 b E_{\alpha_j,a}(b)+o(1).
\end{equation}
Put~\eqref{eq:f6} in~\eqref{eq:f5} to get the following  upper bound 
\begin{equation}\label{eq:f7}
\mathcal E_0(\psi,\Ab;\mathcal N_j(\hat\ell))\leq\frac 1 b E_{\alpha_j,a}(b)+o(1).
\end{equation}
\paragraph{\emph {The order parameter estimates}}
Integrating by parts in the first equation of~\eqref{eq:Euler}, we get (see~\cite[(6.2)]{fournais2011nucleation})
\begin{equation*}\label{eq:f_psi}
\int_{\mathcal N_j(\hat{\ell})}\left(\big|(\nb-i \kp H \Ab)f_j\psi\big|^2 -|\nabla f_j|^2|\psi|^2 \right)\,dx=\kappa^2\int_{\mathcal N_j(\hat{\ell})}\left(|\psi|^2-|\psi|^4\right)f_j^2\,dx\,.
\end{equation*}
Consequently,
\begin{align*}\label{eq:f_psi_1}
\mathcal E_0\big(f_j\psi,\Ab;\mathcal N_j(\hat{\ell})\big)&=\kappa^2\int_{\mathcal N_j(\hat{\ell})}f_j^2\Big(-1+\frac 12 f_j^2\Big)|\psi|^4\,dx+\int_{\mathcal N_j(\hat{\ell})}|\nabla f_j|^2|\psi|^2\,dx\nonumber\\
&=-\frac 12\kappa^2\int_{\mathcal N_j({\ell})}|\psi|^4\,dx+\kappa^2\int_{\mathcal N_j(\hat{\ell})\setminus \mathcal N_j({\ell})}f_j^2\Big(-1+\frac 12 f_j^2\Big)|\psi|^4\,dx\nonumber\\&\quad+\int_{\mathcal N_j(\hat{\ell})}|\nabla f_j|^2|\psi|^2\,dx.
\end{align*}
Again, using the exponential decay of $\psi$ in $\mathcal N_j(\hat{\ell})\setminus \mathcal N_j({\ell})$, and the properties of $f_j$ in~\eqref{eq:f}, we get 
\[\mathcal E_0\big(f_j\psi,\Ab;\mathcal N_j(\hat{\ell})\big)=-\frac 12\kappa^2\int_{\mathcal N_j({\ell})}|\psi|^4\,dx+o(1).\]
Implement this equation in~\eqref{eq:f3} and use the bounds in~\eqref{eq:f4} and~\eqref{eq:f7} to get the estimates of the $L^4$-norm of $\psi$ in~\eqref{eq:m2}.
\paragraph{\emph{The global energy estimates}} Now, the  estimates of the ground-state energy in~\eqref{eq:m3} are easy to establish.
First notice that 
\begin{align*}
\Es&= \mathcal E_0(\psi,\Ab;\Om)+\kp^2H^2\|\curl(\Ab-\Fb)\|_{L^2(\Om)}^2\nonumber\\
&=\mathcal E_0\big(\psi,\Ab;\mathcal S_{\hat\ell}\big)+\kp^2H^2\|\curl(\Ab-\Fb)\|_{L^2(\Om)}^2+o(1),\nonumber\\
&=\frac 1 b \sum_{j\in T}E_{\alpha_j,a}(b)+\kp^2H^2\|\curl(\Ab-\Fb)\|_{L^2(\Om)}^2+o(1),
\end{align*}
where $\mathcal S_{\hat\ell}:=\bigcup_{j\in T}\mathcal N_j(\hat{\ell})$. The equality above is obtained by summing over $T$ in~\eqref{eq:f4} and~\eqref{eq:f7}, and due to the aforementioned decay of the minimizer and its energy in $\mathcal S_{\hat\ell}^\complement$. Moreover, this decay and  Item~3 in Theorem~\ref{thm:priori} assure that
\begin{equation*}
\kp^2H^2\|\curl(\Ab-\Fb)\|_{L^2(\Om)}^2\leq C\kappa^2\|\psi\|_{L^2(\Om)}^4 \leq \tilde C\kappa^2\ell^4=o(1).
\end{equation*}
This completes the proof of~\eqref{eq:m3}.
\end{proof}
 \section*{Acknowledgements} I would like to thank Ayman Kachmar for his valuable comments about this article.

\end{document}